\documentclass[journal,12pt,draftclsnofoot,onecolumn]{IEEEtran}
\ifCLASSINFOpdf
\else
\fi
\usepackage{caption}
\usepackage{epsfig}
\usepackage{amsthm,amssymb,graphicx,graphicx,multirow,amsmath,color,mathtools}
\usepackage{tikz}
\usepackage{pgfplots}
\usepackage{circuitikz}
\newcommand{\ignore}[1]{}
\usepackage{epstopdf}
\usetikzlibrary{circuits.logic.US,calc}
\usetikzlibrary{patterns}
\usepackage{verbatim}
\usepackage{psfrag}
\usepackage{bigints}

\def\E{\mathbb{E}}

\def\P{\mathbb{P}}

\def\ie{{\em i.e.}}

\def\d{\mathrm{d}}

\def\sir{\mathtt{SIR}}

{}
  
  \newtheorem{theorem}{Theorem}

\usepackage{subcaption}
\title{EVM analysis of  an Interference Limited SIMO-SC System With Independent and Correlated Channels}
\author{\IEEEauthorblockN{Sudharsan Parthasarathy$^{(a)}$, Suman Kumar$^{(b)}$, Sheetal Kalyani$^{(a)}$}\\
\IEEEauthorblockA{
${}^{(a)}$Dept. of Electrical Eng., Indian Institute of Technology Madras, Chennai, India\\
${}^{(b)}$Dept. of Electrical Eng., Indian Institute of Technology Ropar, Ropar, India\\
\{sudharsan.p, skalyani\}@ee.iitm.ac.in, suman@iitrpr.ac.in} 
}
\begin{document}
\maketitle
\begin{abstract} 
In this paper, we derive the error vector magnitude (EVM) in  a selection combining (SC) system experiencing co-channel interference, for arbitrary number of antennas and interferers  when all the channels experience Rayleigh fading. We use a novel approach that uses the CCDF of SIR to derive EVM as using the conventional approach to derive EVM for a SC system is difficult. Considering two selection rules based on (a) maximum signal power (b) maximum signal to interference ratio, we observe that EVM is worse when maximum signal power based rule is used. Further, EVM is also derived considering (a) all the channels to be independent and (b) channels to be correlated due to insufficient antenna spacing at receiver. For some special cases, EVM is also derived when the desired channels experience Nakagami-$m$ fading.
\end{abstract}
\section{Introduction}
\label{intro}  
Error vector magnitude (EVM) is an alternate performance metric along with bit error rate (BER), coverage probability, throughput etc. for analyzing the performance of wireless systems and is being increasingly used in industry and academia  \cite{ wu2017error, ni2017minimizing}. System impairments such as local oscillator frequency error, phase noise, IQ imbalance, non-linearity etc. have been identified using EVM  \cite{Georgiadis2004}. EVM-based specifications have also been accepted as a part of the IEEE 802.11 family of Wireless Local Area
Network (WLAN)  and Wideband Code Division Multiple Access (W-
CDMA) standards \cite{thomaserror}. The performance of wireless systems are mainly limited by fading and interference. Both the limitations can be mitigated by employing multiple antennas at the receiver. Selection combining (SC) is a simple low complexity receiver diversity scheme. In a SC system, the antenna corresponding to the maximum signal power is chosen in the absence of interference. In the presence of interference, the antenna can be chosen either based on maximizing the signal to interference ratio (SIR) or only signal power \cite{yang2006performance}. Considering both these scenarios, we have derived and studied the EVM in a SC system.

EVM was first derived for digital communication systems in \cite{Gharaibeh2004}.  In \cite{Mahmoud2009}, EVM for a wireless SISO system with Rayleigh fading was derived without considering interference.  EVM for a single-input-multiple-output (SIMO)  optimal maximal ratio combining  without co-channel interference was obtained in \cite{thomaserror}. EVM in a single-input-single-output system with co-channel interference was first obtained in \cite{sudharsanevm}. To the best of our knowledge, there has been no work in literature that has derived EVM in a SIMO-SC system even in the absence of interference.  In this work, we have derived the EVM in a SIMO-SC interference limited system with arbitrary number of antennas and interferers when all the channels experience independent and identically distributed (i.i.d) Rayleigh fading. The analysis is done for both maximum signal power and maximum signal to interference ratio (SIR) based SC rules and we observe SIR based approach gives lower EVM as expected. We employ a novel approach that uses the complementary cumulative distribution function (CCDF) of SIR to derive EVM, since the conventional approach to derive EVM for a SC system is difficult. The derived expressions are in terms of elementary functions and hence can be easily evaluated.  We have also extended the derivation for the case when the desired channel experiences Nakagami-$m$ fading.

In practice, wireless channels are not independent as assumed. One of the reasons for channels to be correlated is the insufficient antenna spacing at the receiver. A thumb of rule in antenna design is that the distance between antennas at a receiver should be atleast half the wavelength for the channels observed at antennas to be independent. Despite decrease in wavelength due to mm-wave, with the advent of massive-MIMO technology, it will be difficult to maintain enough spacing between antennas for channels observed at antennas to be independent. This has motivated us to derive EVM for some special cases considering correlated channels. We observe that EVM increases due to correlation among channels.

\section{System Model}
We consider a SIMO interference limited system where all the transmitters have one antenna each, receiver has $L$ antennas. There are $M$ interferers considered. The symbol at $l$-th receive antenna at $i$-th time instant is \[y_l(i)=h_{l,0} D_0(i)+ \sum\limits_{j=1}^M  h_{l,j} I_j(i)\] where $h_{l,0}$ is the desired channel response between the transmit and $l$-th receive antenna, $D_0(i)$ is the desired symbol transmitted at $i$-th time slot, $h_{l,j}$ is the $j$-th interfering channel's response between the transmit and $l$-th receive antenna, $I_j(i)$ is the symbol from $j$-th antenna transmitted at $i$-th time slot.  There are $N$ slots in a block, $i=1,2,..N$ for which the channels are constant. We assume the symbols transmitted are of average unit energy and zero mean.   We denote the desired channel corresponding to the selected antenna as $h_0^{\prime}$ and the $j$-th interfering channel corresponding to the selected antenna is $h_j^{\prime}$. Let $y^{\prime}(i)$ be the output of the strongest receive antenna in the $i$-th time slot, \ie, $y^{\prime}(i)=h_0^{\prime} D_0(i)+ \sum\limits_{j=1}^M   h_j^{\prime} I_j(i).$ When the symbols are of average unit energy, the EVM   is defined in \cite{thomaserror} as
\ifCLASSOPTIONtwocolumn
\[\text{EVM}=\E \left(\sqrt{\frac{ \sum\limits_{i=1}^N |\frac{y^{\prime}(i)}{h_{0}^{\prime}} -D_0(i)|^2 }{N}}  \right).\] 
\else
\[\text{EVM}=\E \left(\sqrt{\frac{ \sum\limits_{i=1}^N |\frac{y^{\prime}(i)}{h_{0}^{\prime}} -D_0(i)|^2 }{N}}  \right).\]
\fi
Substituting for $y^{\prime}(i)$, and following the same steps as equations (3)-(5) in \cite{sudharsanevm},
\begin{equation}
\text{EVM}= \E \left( \sqrt{\frac{\sum\limits_{j=1}^M  |h_j^{\prime}|^2 }{|h_0^{\prime}|^2}   } \right).
\label{eqn:evm}
\end{equation}
\section{EVM for Independent channels}
\subsection{Rayleigh faded channels}
In this Section, we will derive the EVM when all the channels experience i.i.d. Rayleigh fading. First, we will derive the EVM when selection rule is based on maximum SIR.
\begin{theorem}
\label{thm:sir}
In an interference limited SIMO-SC system of $L$ receive antennas with $M$ interferers where selection combining rule is based on maximum SIR and all the channels experience i.i.d Rayleigh fading of unit mean power, EVM is given by
\[   \sqrt{\pi} \sum\limits_{k=1}^L (-1)^{k-1} \binom{L}{k} \frac{\Gamma(0.5+k M)}{\Gamma(k  M)}. \]
\end{theorem}
\begin{proof}
Let $\text{SIR}^{\prime}=\frac{|h_{0}^{\prime}|^2 }{\sum\limits_{j=1}^M |h_j^{\prime}|^2}$. From \eqref{eqn:evm},
\begin{equation}
\text{EVM}=\int\limits_0^{\infty} \frac{1}{\sqrt{x}} f_{\text{SIR}^{\prime}}(x) \d x.
\label{eqn:evm1}
\end{equation}

In maximizing SIR SC rule, $\text{SIR}^{\prime}$ is the maximum of SIRs in all antennas, \ie, $\text{SIR}^{\prime}$=$\max(\text{SIR}_1,...,\text{SIR}_L)$. The cumulative distribution function (CDF) is
\begin{align}
F_{\text{SIR}^{\prime}}(x)
&= \P(\max(\text{SIR}_1,\cdots,\text{SIR}_L) \leq x) \nonumber\\
&\stackrel{(a)}= \prod\limits_{i=1}^L F_{\text{SIR}_i}(x) ,
\label{eqn:FSIRprime}
\end{align}
where $F_{\text{SIR}_i}(x) = \P(\text{SIR}_i \leq x)$ and  $(a)$ is due to independence of channels. 
PDF of $\text{SIR}^{\prime}$ can be derived as 
\begin{align}
f_{\text{SIR}^{\prime}}(x) &= \frac{\partial}{\partial x} F_{\text{SIR}^{\prime}}(x) =\sum\limits_{i=1}^L f_{\text{SIR}_i}(x) \prod\limits_{j=1 \neq i}^L F_{\text{SIR}_j}(x).
\label{eqn:fSIRprime}
\end{align}
To compute EVM, \eqref{eqn:fSIRprime} will be substituted in \eqref{eqn:evm1}. But it is difficult to analytically compute EVM for arbitrary number of interferers and antennas using this approach. Hence we use a slightly different approach to compute EVM.
From \eqref{eqn:evm1}, $ \text{EVM} =\E \left(\sqrt{\frac{1}{\text{SIR}^{\prime}}} \right).$
For a positive random variable $X$, $\E(X)$=$\int\limits_0^{\infty} \P(X>x) \d x.$
Hence
\begin{align}
\text{EVM} &= \int\limits_0^{\infty} \P \left(\sqrt{\frac{1}{\text{SIR}^{\prime}}}>x \right) \d x = \int\limits_0^{\infty} F_{\text{SIR}^{\prime}}(x^{-2}) \d x.
\label{eqn:EVM_sirprime}
\end{align}
We first compute $F_{\text{SIR}^{\prime}}(x)$ in \eqref{eqn:FSIRprime}.
Let $\text{SIR}_i$=$\frac{X}{Y}$, where $X$=$|h_{i,0}|^2$, $Y$=$\sum\limits_{j=1}^M |h_{i,j}|^2$.
\begin{align}
F_{\text{SIR}_i}(x) &= \P \left(\frac{X}{Y} \leq x \right) = \int\limits_0^{\infty} \P(X \leq x y) f_Y(y) \d y.
\label{eqn:FSIRi}
\end{align}
We assume all the channels to be of unit mean power. Note that if the desired channel is Rayleigh faded, desired channel power $X$ is exponential distributed. Similarly, if the $M$ interfering channels are i.i.d. Rayleigh faded of unit mean power  then the total interference power $Y$ is Gamma distributed of shape and scale parameters $ M$, $1$, respectively.  Thus, $F_X(x)$ is $1-e^{-x}$ and $f_Y(y)$ is $\frac{  y^{ M-1} e^{-y }}{\Gamma( M)}$. Hence from \eqref{eqn:FSIRi},
\begin{align*}
F_{\text{SIR}_i}(x) &= \int\limits_0^{\infty} (1-e^{-x y}) \frac{ y^{ M-1} e^{-y }}{\Gamma(M)} \d y.\\
&= 1-\frac{1}{(1+x)^{ M}}.
\end{align*}
Thus, from \eqref{eqn:FSIRprime},
\begin{equation}
F_{\text{SIR}^{\prime}}(x)= \left(1-\frac{1}{(1+x)^{ M}} \right)^L. 
\label{eqn:FSIRprime1}
\end{equation}
By using Binomial theorem expansion formula $(a+b)^n$=$\sum\limits_{k=0}^n \binom{n}{k} a^k b^{n-k}$, 
\begin{equation}
F_{\text{SIR}^{\prime}}(x) = \sum\limits_{k=0}^L \binom{L}{k} (-1)^k (1+x)^{-kM}.  
\label{eqn:FSIRprime2}
\end{equation}
Substituting \eqref{eqn:FSIRprime2}  in \eqref{eqn:EVM_sirprime} and using Mathematica, EVM is derived.

\end{proof}
Next we will derive the EVM when signal power maximization is the selection combining rule.
\begin{theorem}
\label{thm:signal}
In an interference limited SIMO-SC system of $L$ receive antennas with $M$ interferers where selection combining rule is based on maximum signal power and all the channels experience i.i.d Rayleigh fading of unit mean power, EVM is given by
\[L \sum\limits_{n=0}^{L-1}  	\dbinom{L-1}{L-1-n} (-1)^{n}  \sqrt{\frac{\pi}{n+1}} \frac{ \Gamma(M+0.5)}{ \Gamma( M)}.\]
\end{theorem}
\begin{proof}
Let $g_j^{\prime}$=$|h_j^{\prime}|^2$, $g_{l,j}=|h_{l,j}|^2$.
In selection rule based on maximum signal power, $|h_0^{\prime}|^2$ in \eqref{eqn:evm} is the maximum of desired signal powers received in the $L$ antennas, \ie, $g_0^{\prime}$=$\max(g_{1,0},\cdots, g_{L,0})$. Also, $g_j^{\prime}$=$|h_j^{\prime}|^2$ $\forall j \neq 0$ in \eqref{eqn:evm} is the interfering power from the $j$-th interferer in the antenna corresponding to the maximum signal power. Hence from \eqref{eqn:evm},
\begin{equation}
 \text{EVM} = \int\limits_0^{\infty} \int\limits_0^{\infty} \sqrt{\frac{y}{x}} f_I(y) f_{g_0^{\prime}}(x) \d x \d y,
 \label{eqn:EVM_signalmax}
 \end{equation}
where $I=\sum\limits_{j=1}^M g_j^{\prime}$.
To derive $f_{g_0^{\prime}}(x)$, we will first derive the CDF of $g_0^{\prime}$. Note,
\begin{equation}
 F_{g_0^{\prime}}( x) = F(g_{1,0} \leq x, \cdots, g_{L,0} \leq x). 
 \label{eqn:cdf_g0prime}
 \end{equation}
As the desired channels experience i.i.d. Rayleigh fading of unit mean power, $F_{g_{l,0}} (x)=1-e^{-x}$, $\forall$ $l=1,\cdots,L$. Thus, $F_{g_0^{\prime}}(x) =(1-e^{-x})^L$, $f_{g_0^{\prime}}(x)  \overset{\Delta}{=} \frac{\partial}{\partial x} F_{g_0^{\prime}}(x)=L(1-e^{-x})^{L-1} e^{-x}$ and
\begin{align}
\int\limits_0^{\infty} \sqrt{\frac{1}{x}} f_{g_0^{\prime}}(x) \d x &= \int\limits_0^{\infty} \sqrt{\frac{1}{x}} L(1-e^{-x})^{L-1} e^{-x} \d x \nonumber \\
&\stackrel{(a)}=  L \int\limits_{0}^{\infty}\sqrt{\frac{1}{x }}  \sum\limits_{n=0}^{L-1}  \frac{(-1)^{n}  	\dbinom{L-1}{L-1-n}}{ (e^{x})^{n+1}} \d x \nonumber \\
&\stackrel{(b)}= L \sum\limits_{n=0}^{L-1}  	\dbinom{L-1}{L-1-n} (-1)^{n}  \sqrt{\frac{\pi}{n+1}},
\label{eqn:evmx}
\end{align}
(a) as $(1-e^{-x})^{L-1}=\sum\limits_{n=0}^{L-1}  	\dbinom{L-1}{L-1-n} (-e^{-x})^{n},$ (b) using the identity $\int\limits_{0}^{\infty} \frac{(e^{-x})^{n+1}}{\sqrt{x}} \d x = \sqrt{\frac{\pi}{n+1}}$. 
As the interferers are i.i.d. Nakagami-$m$ faded of unit mean power, probability density function (pdf) of $I$=$\sum\limits_{j=1}^M g_j^{\prime}$ is
$f_{I}(y)= \frac{y^{M-1} e^{-y}}{\Gamma( M)}.$
Hence 
\begin{align}
 \int\limits_0^{\infty} \sqrt{y} f_I(y) \d y &= \frac{1}{\Gamma(M)}   \int\limits_0^{\infty} e^{-y} y^{ M-0.5} \d y \nonumber  \\ 
 &=\frac{ \Gamma(M+0.5)}{ \Gamma( M)} .
 \label{eqn:evmy}
\end{align}
Multiplying \eqref{eqn:evmx} and \eqref{eqn:evmy}, EVM is obtained.
\end{proof}
Next we derive EVM for some special cases when the desired channel experiences Nakagami-$m$ fading. 
\subsection{Desired channel is Nakagami-$m$ faded}
First we derive below the EVM for selection rule based on maximum SIR.
\begin{theorem}
\label{thm:nakagami_sir}
In an interference limited SIMO-SC system of $L$ receive antennas where selection combining rule is based on maximum SIR with $L$ desired channels experiencing Nakagami-$m$ fading and channels from 2 interferers experiencing i.i.d. Rayleigh fading , EVM is given by
\begin{equation*}
\frac{\Gamma(-\frac{1}{2}+L+L m_d) {}_2F_1(\frac{1}{2},-L,\frac{3}{2}-L(1+m_d),1+m_d)}{2 (m_d \pi)^{-0.5} \Gamma(L+L m_d)  }+ 
 \end{equation*}
\begin{equation*}
\frac{ \Gamma(\frac{1}{2}-L (1+m_d)) \Gamma(L m_d-\frac{1}{2}) }{2 \Gamma[-L]  (m_d )^{-0.5} (1+m_d)^{0.5-L-L m_d} } \times
\end{equation*}
\begin{equation*}
{}_2F_1(L(1+m_d),L m_d-\frac{1}{2},\frac{1}{2}+L(1+m_d) ,1+m_d)
\end{equation*}
\end{theorem}
\begin{proof}
To compute EVM in \eqref{eqn:EVM_sirprime}, we need $F_{\text{SIR}^{\prime}}(x)$ in \eqref{eqn:FSIRprime}, for which $F_{\text{SIR}_i}(x)$ is derived first.  
We assume all the channels to be of unit mean power. Note that if the desired channel is Nakagami-$m$ faded, desired channel power $X$ is Gamma distributed of shape parameter $m_d$. Similarly, if the 2 interfering channels are i.i.d. Rayleigh faded of unit mean power then the total interference power $Y$ is Gamma distributed of shape and scale parameters $2$, $1$ respectively.  Thus, $F_X(x)$ is $1-\frac{\Gamma(m_d,m_d x)}{\Gamma(m_d)}$ and $f_Y(y)$ is $  y e^{-y }$. Hence from \eqref{eqn:FSIRi},
\begin{align*}
F_{\text{SIR}_i}(x) &= \int\limits_0^{\infty} \left(1-\frac{\Gamma(m_d,m_d x y)}{\Gamma(m_d)} \right)y e^{-y} \d y. \\
&=\frac{(m_d x)^{m_d} (1+m_d+m_d x) }{(1+m_d x)^{1+m_d}}.
\end{align*}
Thus, from \eqref{eqn:FSIRprime},
\begin{equation}
F_{\text{SIR}^{\prime}}(x) = \left(\frac{(m_d x)^{m_d} (1+m_d+m_d x) }{(1+m_d x)^{1+m_d}}\right)^L.
\label{eqn:fSIRprime_rayl}
\end{equation}
Substituting \eqref{eqn:fSIRprime_rayl} in \eqref{eqn:EVM_sirprime}, EVM is derived. 
\end{proof}
\hspace{-5mm} Now we derive EVM when signal power maximization is used. 
\begin{theorem}
\label{thm:signal_rayl}
In an interference limited SIMO-SC system of 2 receive antennas where selection combining rule is based on maximum signal power with the 2 desired channels experiencing Nakagami-$m$ fading and channels from $M$ interferers experiencing i.i.d. Rayleigh fading , EVM is given by
\[  \frac{2  \Gamma(m_d-\frac{1}{2}) \left(1- \frac{{}_2F_1(m_d-\frac{1}{2},2 m_d-\frac{1}{2}, m_d+\frac{1}{2},-1)}{\Gamma(m_d)\Gamma(m_d+0.5) \Gamma(2 m_d-0.5)^{-1}}\right) \Gamma( M+\frac{1}{2}) }{m_d^{-\frac{1}{2}}\Gamma(m_d) \Gamma( M) }\]
\end{theorem}
\begin{proof}
To derive EVM in \eqref{eqn:EVM_signalmax}, pdf of $g_0^{\prime}$ will be derived. We consider two antennas. So $g_0^{\prime}$ is max of independent channels $g_{1,0}$ and $g_{2,0}$. Both these channels are Nakagami-$m$ faded of unit mean power and shape parameter $m_d$. Hence the CDF of $g_0^{\prime}$  from \eqref{eqn:cdf_g0prime} is
$F_{g_0^{\prime}}(x)= \left(\frac{\Gamma(m_d)-\Gamma(m_d,m_d x)}{\Gamma(m_d)} \right)^2$ and the pdf is \[f_{g_0^{\prime}}(x)  \overset{\Delta}{=} \frac{\partial}{\partial x} F_{g_0^{\prime}}(x)=\frac{2(\Gamma(m_d)-\Gamma(m_d,m_d x))e^{-m_d x} x^{m_d-1}}{\Gamma(m_d)^2 m_d^{-m_d}}. \]
Thus, $\int\limits_0^{\infty} \sqrt{\frac{1}{x}} f_{g_0^{\prime}}(x) \d x$ is 
\begin{equation}
 \frac{2  \Gamma(m_d-\frac{1}{2}) }{m_d^{-\frac{1}{2}}\Gamma(m_d)} \left(1- \frac{{}_2F_1(m_d-\frac{1}{2},2 m_d-\frac{1}{2}, m_d+\frac{1}{2},-1)}{\Gamma(m_d)\Gamma(m_d+0.5) \Gamma(2 m_d-0.5)^{-1}}\right).
 \label{eqn:g0prime_rayl}
\end{equation} 
Similarly $\int\limits_0^{\infty} \sqrt{y} f_I(y) \d y$ can be found from  \eqref{eqn:evmy}. Multiplying \eqref{eqn:g0prime_rayl} and \eqref{eqn:evmy}, EVM is obtained.
\end{proof}
In the next Section, EVM is derived when the channels at receive antennas are positively correlated and all the channels experience Rayleigh fading.
\section{EVM-Correlated receive antennas}
First, we derive EVM when selection rule is based on maximum SIR.
\begin{theorem}
\label{thm:SIRmaxcorr}
EVM of an interference limited two antenna correlated system with one interferer  using maximum SIR based selection rule, when all the channels experience Rayleigh fading, is given as 
  \[\int\limits_0^{\infty}\frac{1}{1+x^{2}} - \frac{(1+x^{-1})^{-1} x^{2}(\sqrt{1-\rho^2})}{  \sqrt{1-\rho^2 + 2 (1+\rho^2) x^{2} +(1-\rho^2) x^{4}} } \d x, \]
  where $\rho$ is the correlation coefficient.
\end{theorem}
\begin{proof}
EVM of an interference limited selection combining system where maximum SIR based selection rule is applied is given in \eqref{eqn:EVM_sirprime}.
In (6) in \cite{okui}, closed form expression for CDF of $\sir^{\prime}$ in \eqref{eqn:FSIRprime} was derived for a two antenna correlated system with one interferer when all the channels experience Rayleigh fading as
\[F_{\text{SIR}^{\prime}}(x)=\frac{1}{1+x^{-1}} - \frac{(1+x^{-1})^{-1} x^{-1}(\sqrt{1-\rho^2})}{  \sqrt{1-\rho^2 + 2 (1+\rho^2) x^{-1} +(1-\rho^2) x^{-2}} } .\]
Substituting this in \eqref{eqn:EVM_sirprime}, EVM is derived. Closed form expression cannot be obtained, but this integral can be numerically evaluated. 
\end{proof}
Next we derive EVM when maximum signal power based selection combining rule is used. 
\begin{theorem}
\label{thm:signalmaxcorr}

EVM of an interference limited two antenna correlated system with $M$ interferers  using maximum signal power based selection rule, when all the channels experience Rayleigh fading,  is given as 
\[\frac{\Gamma(M+0.5)}{\Gamma(M)}\int\limits_0^{\infty} \frac{2 e^{-x} (1-Q \left(\rho\sqrt{\frac{2x}{1-\rho^2}} , \sqrt{\frac{2x}{1-\rho^2}} \right) )}{\sqrt{x}}  \d x.\]
\end{theorem}
\begin{proof}
EVM of an interference limited selection combining system where maximum signal power based selection rule is applied is given in \eqref{eqn:EVM_signalmax}. 

First we will derive $\int\limits_0^{\infty} \sqrt{y}f_I(y) \d y$.

As the $M$ interfering channels are considered to be Rayleigh faded of unit mean power, the total interferernce power is Nakagami-$m$ faded of shape and scale parameters $M$  and $1$. In \eqref{eqn:evmy}, we have already derived $\int\limits_0^{\infty} \sqrt{y}f_I(y) \d y = \frac{\Gamma(M+0.5)}{\Gamma(M)}.$
Next,we derive $\int\limits_0^{\infty} \sqrt{\frac{1}{x}}f_{g_0^{\prime}}(x) \d x$.
In  \cite{chen2004distribution}, closed form expression for  PDF of $g_0^{\prime}$ in \eqref{eqn:EVM_signalmax} was derived for a two antenna correlated system when all the channels experience Rayleigh fading as
\[f_{g_0^{\prime}}(x)=2 e^{-x} \left(1-Q \left(\rho\sqrt{\frac{2x}{1-\rho^2}} , \sqrt{\frac{2x}{1-\rho^2}} \right) \right) .\]
Hence
\begin{equation}
\int\limits_0^{\infty} \sqrt{\frac{1}{x}}f_{g_0^{\prime}}(x) \d x = \int\limits_0^{\infty} \frac{2 e^{-x} (1-Q \left(\rho\sqrt{\frac{2x}{1-\rho^2}} , \sqrt{\frac{2x}{1-\rho^2}} \right) )}{\sqrt{x}}  \d x. 
 \label{eqn:evmx_correlation}
 \end{equation}
Deriving closed form expression is not possible, but this integral can be easily evaluated numerically. 
Multiplying  \eqref{eqn:evmy} and \eqref{eqn:evmx_correlation}, EVM is derived.

\end{proof}
Fully correlated case:

When $\rho$=1, there is no selection required as both the antennas receive same signal. Hence EVM can be derived from Theorem \ref{thm:sir} or \ref{thm:signal} by substituting $L$=1. EVM thus derived is \[\text{EVM}=\sqrt{\pi} \frac{\Gamma(M+0.5)}{\Gamma(M)}.\]
By using the identity,
 \begin{equation}
\frac{\Gamma(n+a)}{\Gamma(n+b)}=\frac{(1+\frac{(a-b) (a+b-1)}{2 n} +O(\frac{1}{n^2}))}{n^{b-a}},  \text{for large n},
\label{eqn:identity}
\end{equation}
$\text{EVM} \rightarrow \sqrt{\pi M}$, for large number of interferers. 
\section{Results}

In Fig. \ref{fig:fig1}, EVM is plotted for the two interferer case when desired channel experiences Nakagami-$m$ fading, interferers experience Rayleigh fading and the antenna chosen in selection combining is based on SIR maximization. The analytical result in Theorem \ref{thm:nakagami_sir} is verified by simulation. With increase in $m_d$ (LOS component of the channel),  EVM decreases as expected. The steep decrease in EVM when number of antennas increases from one to three for $m_d$=1 shows multiple antennas combat severe fading very well and the plot also shows diminishing returns with increase in number of antennas.  

In Fig. \ref{fig:fig3}, EVM is plotted for the two antenna case when desired channel experiences Nakagami-$m$ fading where the antenna being chosen in selection combining is based on signal power maximization. The analytical results derived in Theorem \ref{thm:signal_rayl} is verified by simulation. EVM increases with decrease in $m_d$ (LOS component of desired channel) and increase in number of interferers as expected.

In Fig. \ref{fig:fig2},  EVM is plotted by varying the correlation for both SIR and signal power maximization rule. Theorem \ref{thm:SIRmaxcorr} and Theorem \ref{thm:signalmaxcorr} are verified by simulation. 
As i.i.d channels experience correlation to be 0, Theorem \ref{thm:sir} and Theorem \ref{thm:signal} are also verified by simulation in Fig. \ref{fig:fig2} at $\rho$=0. We also observe that SIR maximization performs better than signal power maximization \ie, lower EVM is obtained due to SIR maximization. We observe that EVM increases due to correlated channels at receive antennas.

\ifCLASSOPTIONtwocolumn
\begin{figure}
\includegraphics[scale=0.71]{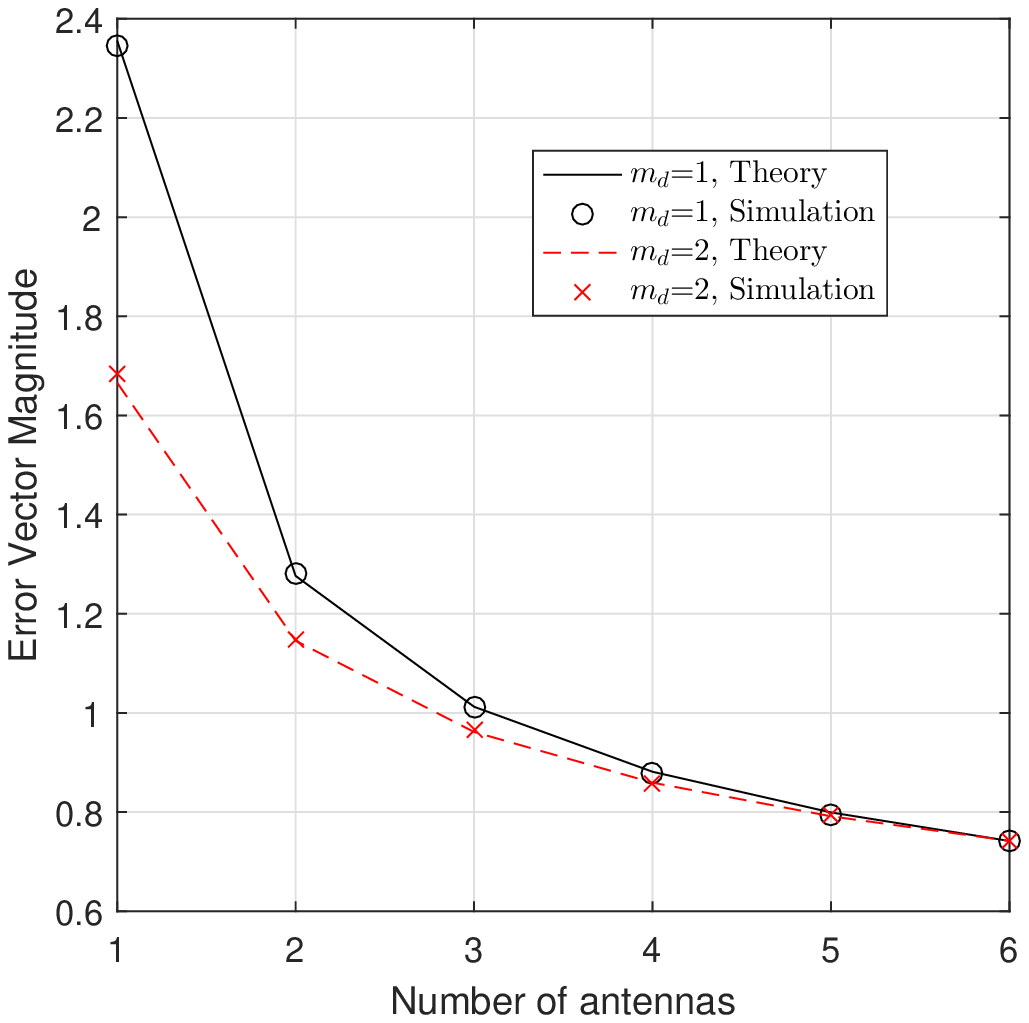}
\centering
\caption{EVM in 2 interferer case, Nakagami-$m$ fading desired channel, Rayleigh fading interferer channels, SIR max. SC rule}
\label{fig:fig1}
\end{figure}
\else
\begin{figure}
\includegraphics[scale=1]{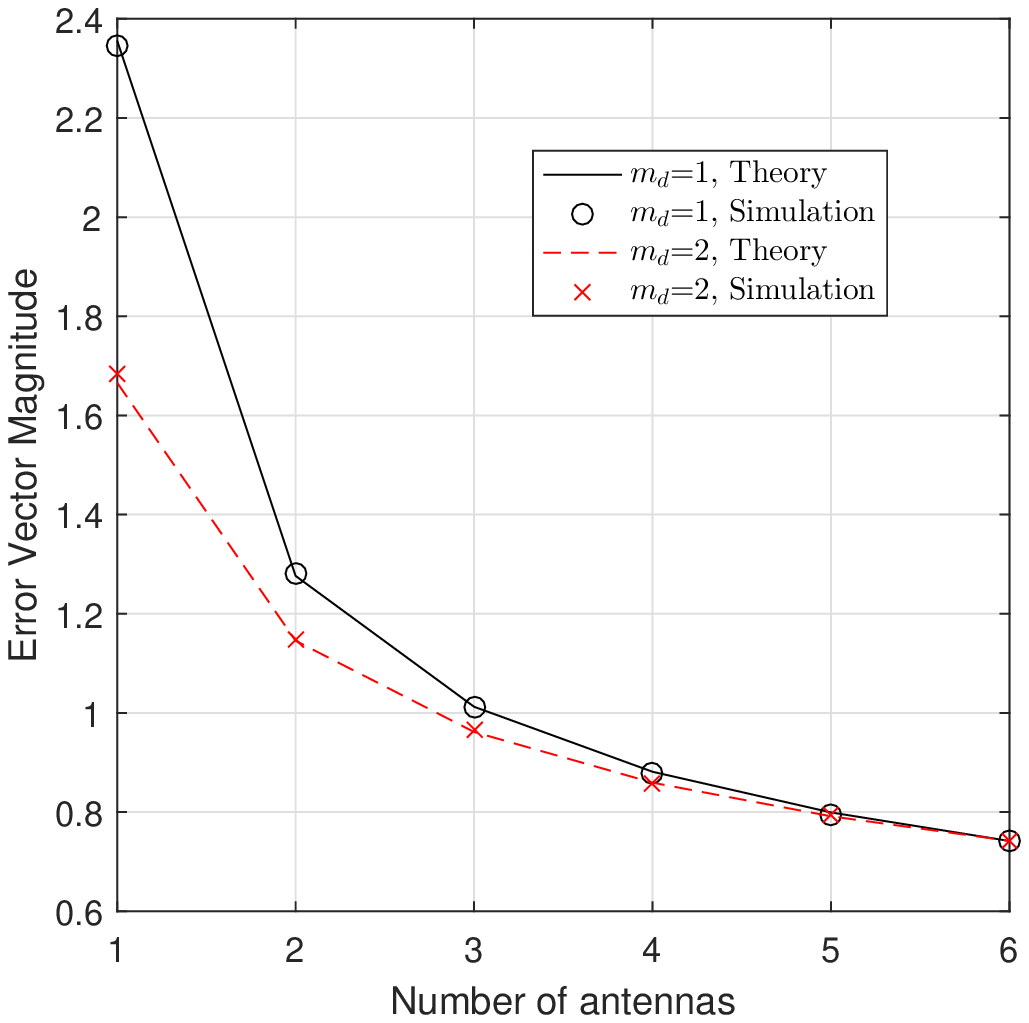}
\centering
\caption{EVM in 2 interferer case, Nakagami-$m$ fading desired channel, Rayleigh fading interferer channels, SIR max. SC rule}
\label{fig:fig1}
\end{figure}
\fi

\ifCLASSOPTIONtwocolumn
\begin{figure}
\includegraphics[scale=0.65]{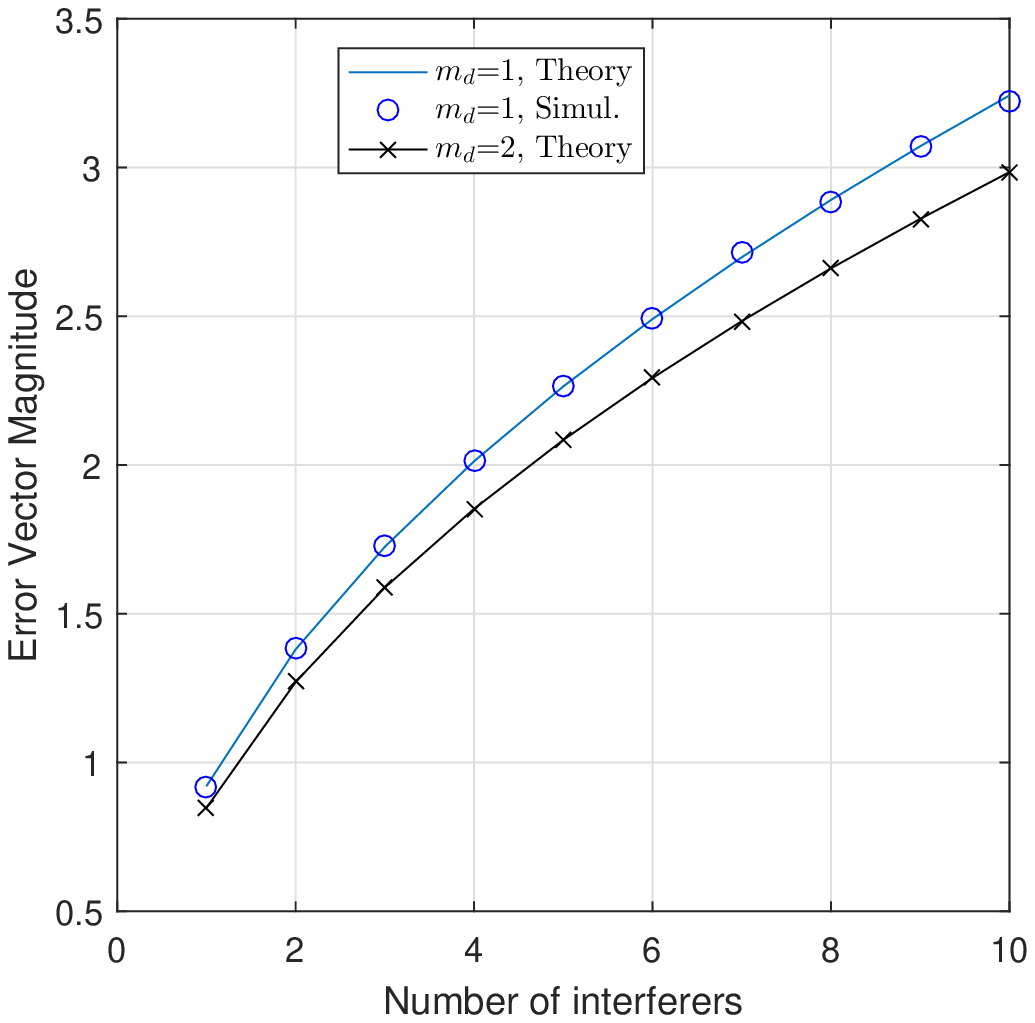}
\centering
\caption{EVM in 2 antenna case with Nakagami-$m$ fading desired and interferer channels, signal maximization SC rule}
\label{fig:fig3}
\end{figure}
\else
\begin{figure}
\includegraphics[scale=.95]{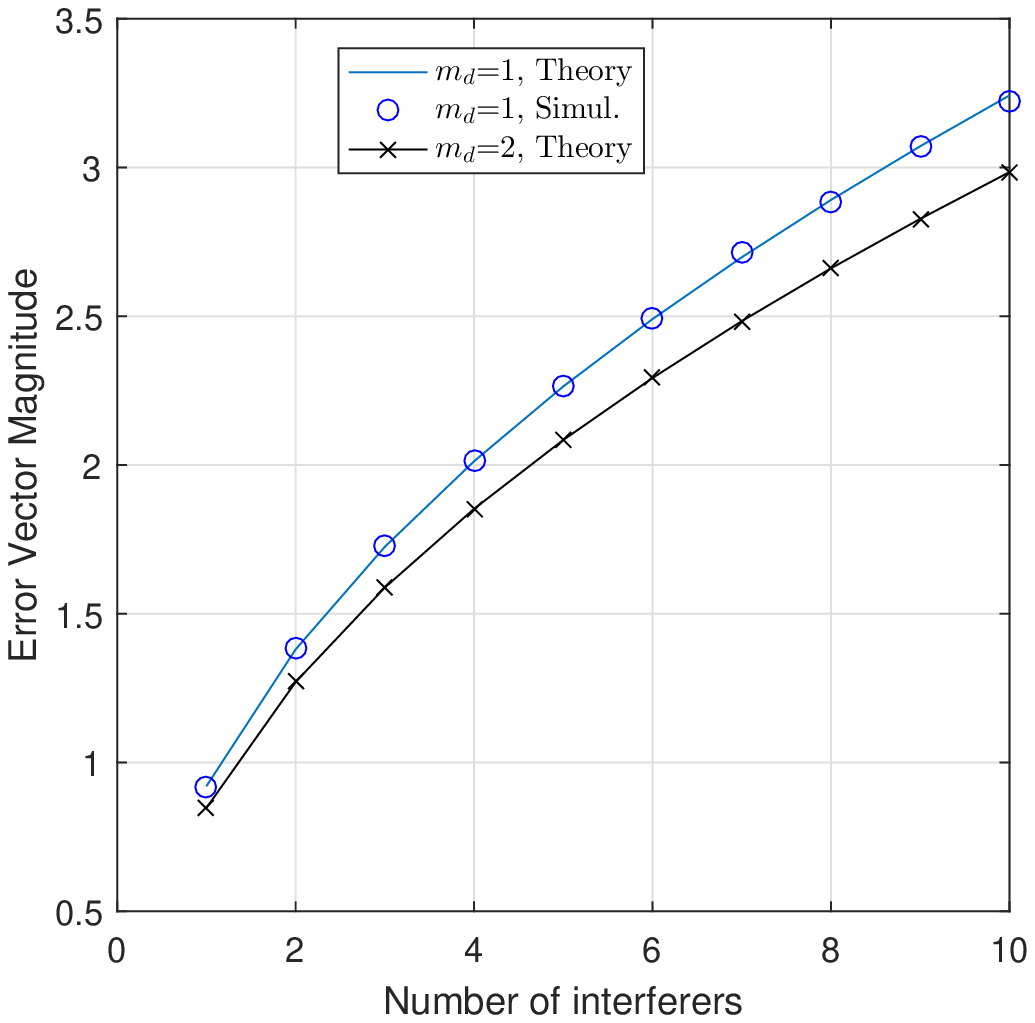}
\centering
\caption{EVM in 2 antenna case with Nakagami-$m$ fading desired and interferer channels, signal maximization SC rule}
\label{fig:fig3}
\end{figure}
\fi

\ifCLASSOPTIONtwocolumn
\begin{figure}
\includegraphics[scale=0.65]{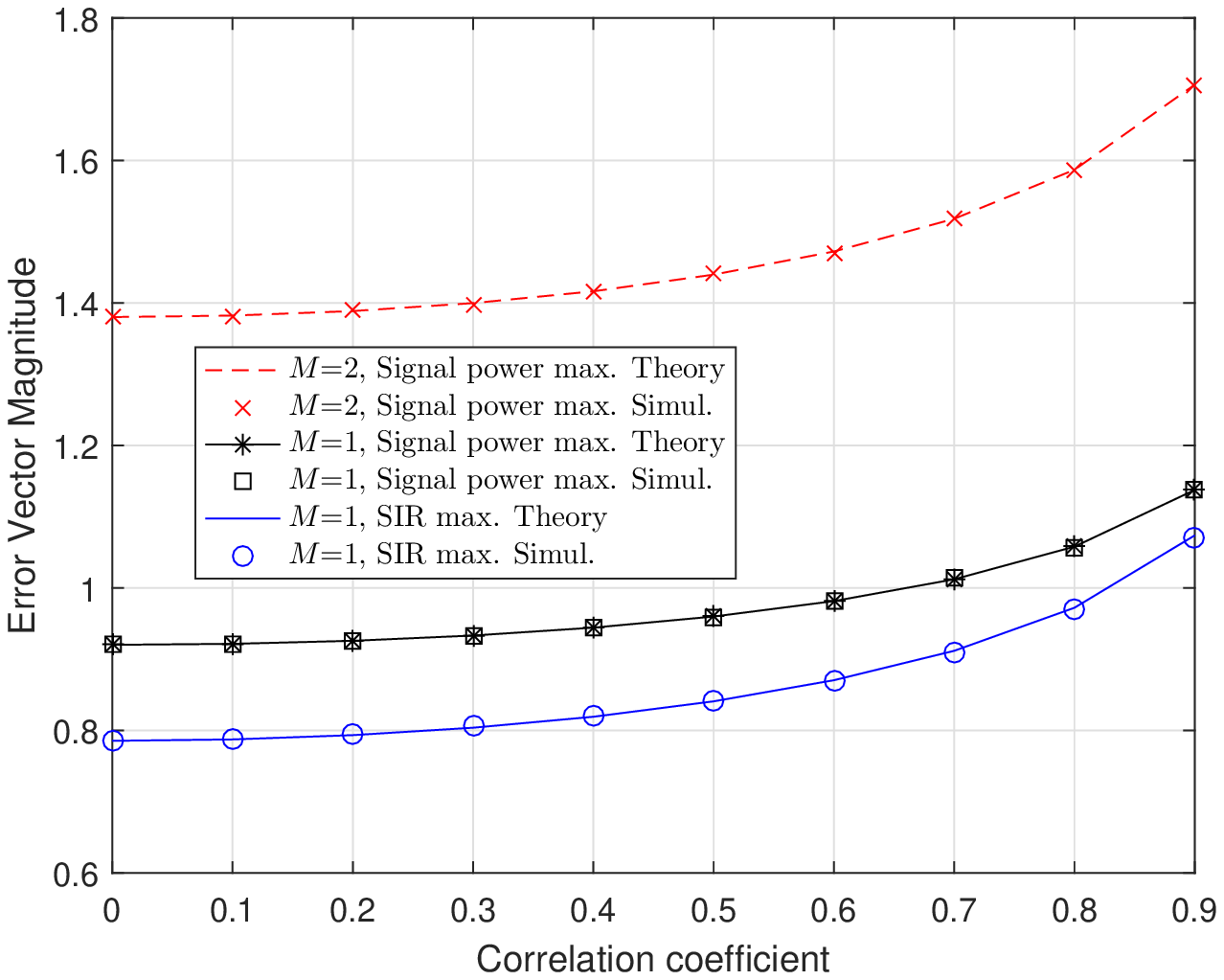}
\centering
\caption{EVM in 2 antenna case with Rayleigh fading in all the channels}
\label{fig:fig2}
\end{figure}
\else
\begin{figure}
\includegraphics[scale=.95]{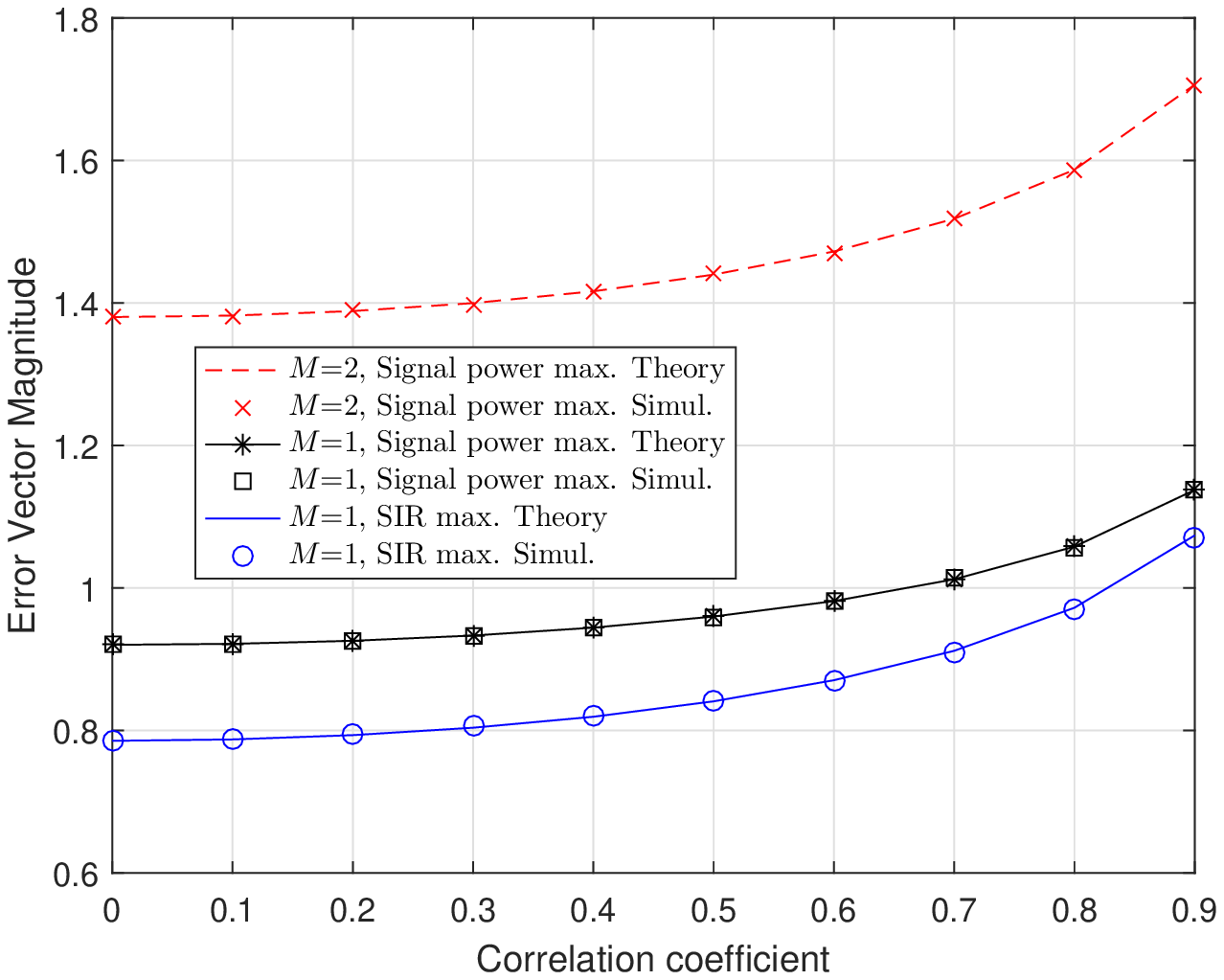}
\centering
\caption{EVM in 2 antenna case with Rayleigh fading in all the channels}
\label{fig:fig2}
\end{figure}
\fi

\section{Conclusion}
In this paper, we have derived the EVM  in an interference limited SIMO-SC system for arbitrary number of antennas and interferers when all the channels experience i.i.d Rayleigh fading. We have used a new approach to derive error vector magnitude and the expressions derived are in terms of elementary Gamma functions. We have compared the performance of two selection combining rules in which the best antenna is chosen by maximizing the signal to interference ratio  and by maximizing only the signal power respectively. We have shown that SIR maximizing performs better than signal power maximizing rule, \ie, the EVM due to SIR maximizing is lesser than signal power maximizing rule. We have also derived EVM when the desired channel experiences Nakagami-$m$ fading for some special cases.
We have also considered the case of correlated receive antennas and shown that EVM increases due to correlation.
The results have also been verified through simulation.

\bibliographystyle{IEEEtran}
\bibliography{References}
\end{document}